\documentclass{llncs}
\usepackage[utf8]{inputenc}
\usepackage{amsmath}
\usepackage{amssymb}
\usepackage{graphicx}
\usepackage{tikz}
\usepackage{caption}
\usepackage{float}
\usepackage{framed}
\usepackage{verbatim}
\usepackage{algorithm}
\usepackage{algpseudocode}
\title{An output-sensitive algorithm for the minimization of 2-dimensional String Covers}

\author{Alexandru Popa\inst{1,2} \and Andrei Tanasescu\inst{3}}

\institute{University of Bucharest \and
National Institute of Research and Development in Informatics \and Politehnica University of Bucharest \\
E-mail: \email{alexandru.popa@fmi.unibuc.ro, andrei.tanasescu@mail.ru}}

\begin{document}
\maketitle{}

\begin{abstract} String covers are a powerful tool for analyzing the quasi-periodicity of 1-dimensional data and find applications in automata theory, computational biology, coding and the analysis of transactional data. A \emph{cover} of a string $T$ is a string $C$ for which every letter of $T$ lies within some occurrence of $C$. String covers have been generalized in many ways, leading to \emph{k-covers}, \emph{$\lambda$-covers}, \emph{approximate covers} and were studied in different contexts such as \emph{indeterminate strings}.

In this paper we generalize string covers to the context of 2-dimensional data, such as images. We show how they can be used for the extraction of textures from images and identification of primitive cells in lattice data. This has interesting applications in image compression, procedural terrain generation and crystallography. 
\end{abstract}

\section{Motivation}

Redundancy is an ubiquitous phenomenon in engineering and computer science ~\cite{ming1990kolmogorov,muchnik2003almost}. Periodicity is the most common and useful form of redundancy. Periodicity is a key phenomenon when analyzing physical data such as an analogue signal. Natural data is very redundant or repetitive and exhibits some patterns or regularities~\cite{HAVLIN1995171,timmermans2017cyclical,tychonoff1935theoremesL} which we may assert to be the intended information~\cite{searle1980speech} within the data. Periodicity itself has been thoroughly studied in various fields such as signal processing~\cite{sethares1999periodicity}, bioinformatics~\cite{brodzik2007quaternionic}, dynamical systems~\cite{katok1997introduction} and control theory~\cite{bacciotti2006liapunov}, each bringing its own insights. 

However, natural data is imperfect. It is highly unlikely that natural data can ever be periodic. In fact, the data is \emph{almost} or \emph{quasi-}periodic~\cite{ApostolicoB97}. This has been firstly studied over strings, the most general representation of digital data~\cite{middlestead2017digital}. 

For example, assume that we  want to send the word $aba$ over a noisy channel as a digital signal where letters are modulated using \emph{amplitude shift keying}~\cite{middlestead2017digital}. Since the simple transmission is unlikely to  yield the result due to the imperfect transmission channel, we add redundancy and thus send the word $aba$ multiple times. However, when errors occur, the received signal only partially retains its periodicity. 

\section{Our Results}

In this paper we study the generalization of the String Cover operator on finite-dimensional images. First, throughout this paper, given two integers $a$ and $b$, $a \le b$, we define $\overline{a,b} = \{a, a+1, \dots, b\}$. Then a 2-dimensional string (or an image) is function $I:\overline{0,H-1}\times\overline{0,W-1}\rightarrow\Sigma$, where $\Sigma$ is the alphabet. Let $Mat_{H,W}$ be the set of all matrices with $H$ rows and $W$ columns. For a matrix $M$ we define $M^i_j$ to be the element from row $i$ and column $j$.

\begin{definition}[2-dimensional string cover] A \emph{cover} of a 2D image $T$ is a 2D image $C$ for which every element of $T$ lies within some occurrence of $C$. 
\end{definition}

In Section~\ref{sec:image_covers}, we find two alternative ways of formalizing the 2D cover problem, by using masks 
and prove their equivalence. We then turn our attention towards the decision problem: 
\begin{problem}[Image Cover Decision] Given two images $T$ and $C$, does one cover the other? 
\end{problem}

We give an $\mathcal{O}\left(WH\right)$ algorithm based upon Bird's 1977~\cite{bird1977two} 2-dimensional matching algorithm. Then, using this algorithm we study the minimization problem (Section~\ref{sec:minimal_image_covers}).

\begin{problem}[Weak Minimal Image Cover] Given an image $T\in\mathrm{Mat}_{H,\,W}\left(\Sigma\right)$ and an evaluation function $eval:\overline{1,\,h}\times\overline{1,\,w}\rightarrow \mathbb{R}$, where $h \le H$ and $w \le W$, which induces an order onto the covers, which is the cover $C\in\mathrm{Mat}_{h,\,w}\left(\Sigma\right)$ of $T$ minimal with respect to $eval\left(h,\,w\right)$?
\end{problem}

We give an $\mathcal{O}\left(W^2H^2\right)\Theta\left(eval\right)$ algorithm. Since the minimization problem is actually $\Omega\left(WH\right)\Theta\left(eval\right)$, we aim for a better algorithm. Using sorting of the input candidates according to $eval$ we obtain $\mathcal{O}\left(nWH\right)\Theta\left(eval\right)$ (the bound does not contain the time necessary for sorting), where the $n$-th entry in the vector sorted by $eval$ determines a cover of $T$. Note that to assume that the candidates are sorted is not in general realistic, so a more honest complexity bound is $\mathcal{O}\left(WH\left(n+\log\left(WH\right)\Theta\left(eval\right)\right)\right)$. However, there is a very important optimization criterion where the sorting is very cheap, namely the \emph{size} of an image.

\begin{problem}[Strong Minimal Image Cover] Given an image $T\in\mathrm{Mat}_{H,\,W}\left(\Sigma\right)$ which is the cover $C\in\mathrm{Mat}_{h,\,w}\left(\Sigma\right)$ of $T$ minimal with respect to its area (that is, $wh$),  $\ell_1$ norm (that is, $w + h$) and $\ell_\infty$ norm (that is, $\max\left(w,\,h\right))$?
\end{problem}

For this problem we augment the general minimization algorithm with a preprocessing routine, based on the optimal 1-dimensional Minimal String Cover algorithm~\cite{ApostolicoFI91}, which reduces the number of candidate pairs that we have to check from $\Theta\left(WH\right)$ to $\mathcal{O}\left(1\right)$ on the average case, reducing the complexity to $\Theta\left(WH\right)$ on the average case and, particularly, $\mathcal{O}\left(W\right)$ in the worst-case for $H=1$. We argue that the use of this routine never hinders performance and offers the same boost for the general case of an unknown $eval$ function.

We conclude the article with a few very interesting applications of other generalizations of the Minimal String Cover Problem (Sections~\ref{sec:lattices} and~\ref{sec:graphics}) such as $k$-covers~\cite{Iliopoulos98} and the Approximate String Cover Problem introduced by Amir et al.~\cite{AmirLLP17,AmirLLLP17} to lattice unit-cell recognition from generic images, detection of the unit cells of some quasicrystals~\cite{wlodawer2013protein}, extraction of the elementary set of tiles in a Wang Tiling, recognizing the minimal (quasi)periodic Wang Tile pattern in an image and the minimal modification required of an image for the existence of a non-trivial minimal (quasi)periodic Wang Tile pattern.

\section{Image Covers}
\label{sec:image_covers}

The simplest class of images is that of binary images, i.e. $\Sigma \cong \lbrace 0,\,1\rbrace$. Binary images can be thought of as sets over $\mathbb{Z}^2$, as follows: the set contains the position (i.e., row and column) of the elements of the binary image that have value $1$.

\begin{example}
The set $\{ (1,2), (2, 2), (3,3)\}$ corresponds to the image
 $$
  \left[ {\begin{array}{ccc}
   0 & 1 & 0 \\
   0 & 1 & 0 \\
   0 & 0 & 1 \\  
\end{array} } \right].
$$
\end{example}

Given a set $S$  and an element $x$, the characteristic function of $S$, denoted by $\chi_S (x)$ has value $1$ if $x \in S$, and $0$ otherwise.

\begin{definition} A mask of an image $T$ with respect to an image $C$ is a binary image $M$ which marks the first position of some occurrences of $C$ in $T$. 

Formally if $T\in\mathrm{Mat}_{H,\,W}\left(\Sigma\right)$ and $C\in\mathrm{Mat}_{h,\,w}\left(\Sigma\right)$ then $M\in\mathrm{Mat}_{H,\,W}\left(\lbrace0,\,1\rbrace\right)$ is a mask of $T$ with respect to $C$ if 
$$\forall i\in\overline{1,\,H},\,j\in\overline{1,\,W}, \, M^i_j=1 \Rightarrow T^{i+y-1}_{j+x-1}=C^y_x\ ,\,y\in\overline{1,\,h},\,x\in\overline{1,\,w}.$$
\end{definition}

By the correspondence between binary images and sets, there exists a maximal mask with respect to cardinality and it identifies all occurrences of an image in another.

\begin{definition} The maximal mask of an image $T$ with respect to an image $C$ is a binary image $M^*$ which marks the first position of all occurrences of $C$ in $T$. 

Formally if $T\in\mathrm{Mat}_{H,\,W}\left(\Sigma\right)$ and $C\in\mathrm{Mat}_{h,\,w}\left(\Sigma\right)$ then $M^*\in\mathrm{Mat}_{H,\,W}\left(\lbrace0,\,1\rbrace\right)$ is the maximal mask of $T$ with respect to $C$ if 
$$\forall i\in\overline{1,\,H},\,j\in\overline{1,\,W},  {M^*}^i_j=1\Longleftrightarrow T^{i+y-1}_{j+x-1}=C^y_x\ ,\,y\in\overline{1,\,h},\,x\in\overline{1,\,w}.$$
\end{definition}

Extrapolating from the definition of string covers, we can informally define a cover of an image. A \emph{cover} of an image $T$ is an image $C$ for which every element of $T$ lies within some occurrence of $C$. We can formalize this definition using masks. We introduce two equivalent definition candidates.

\begin{definition}[Weak Image Covers] 

If $T\in\mathrm{Mat}_{H,\,W}\left(\Sigma\right)$ and $C\in\mathrm{Mat}_{h,\,w}\left(\Sigma\right)$, then $C$ covers $T$ if there exists some mask $M$ of $T$ with respect to $C$ such that:
$$\forall Y\in\overline{1,\,H},\,X\in\overline{1,\,W}\ \exists i\in\overline{Y-h+1,\,Y},\,j\in\overline{X-w+1,\,X}\ M^i_j=1.$$
\end{definition}

Equivalently, we may define Image Covers with respect to the maximal mask:

\begin{definition}[Strong Image Covers] \label{def:strong} 

If $T\in\mathrm{Mat}_{H,\,W}\left(\Sigma\right)$ and $C\in\mathrm{Mat}_{h,\,w}\left(\Sigma\right)$, then $C$ covers $T$ if the maximal mask $M^*$ of $T$ with respect to $C$ is such that:
$$\forall Y\in\overline{1,\,H},\,X\in\overline{1,\,W}\ \exists i\in\overline{Y-h+1,\,Y},\,j\in\overline{X-w+1,\,X}\ {M^*}^i_j=1$$
\end{definition}

By these definitions a cover $C\in\mathrm{Mat}_{h,\,w}\left(\Sigma\right)$ of an image $T\in\mathrm{Mat}_{H,\,W}\left(\Sigma\right)$ can be identified with the $\left(h,\,w\right)$ pair. 

The weak definition is a more natural extension of the definition of String Covers, while the strong definition provides us with a more clear understanding of the combinatorial properties of Image Covers. For example, the strong definition suggests that Image Covers are susceptible to dynamic programming, which we later use to obtain the minimal cover.

\begin{theorem}
The weak and strong definitions are equivalent.
\end{theorem}

\begin{proof}
Consider the set $\mathcal{S}=\overline{1,\,H}\times\overline{1,\,W}$. There exists a bijection between its power set, $\mathcal{P}\left(\mathcal{S}\right)$, and the W-long, H-tall binary images $\mathrm{Mat}_{H,\,W}\left(\lbrace 0,\,1\rbrace\right)$ as explained at the beginning of the section. Formally, the bijection $f$ is defined as 
$$ \mathcal{P}\left(\mathcal{S}\right)\ni S\leftrightarrow f(S)\in\mathrm{Mat}_{H,\,W}\left(\lbrace 0,\,1\rbrace\right):\ {f(S)}^i_j=\chi_{\mathcal{S}}\left(\left(i,\,j\right)\right)\forall i\in\overline{1,\,H},\, j\in\overline{1,\,W}.$$

However, the image of the Boolean algebra $\left(\mathcal{P}\left(\mathcal{S}\right),\,\cup,\,\cap,\,\bar{.},\,\emptyset,\,\mathcal{S}\right)$ is thus by $f$ onto $\mathrm{Mat}_{H,\,W}\left(\lbrace 0,\,1\rbrace\right)$. The new structure can be verified to be $$\left(\mathrm{Mat}_{H,\,W}\left(\lbrace 0,\,1\rbrace\right),\,\max,\,\min,\,\mathbf{M}\rightarrow\mathbf{1}-\mathbf{M},\,\mathbf{0},\,\mathbf{1}\right)$$

Thus the image of the inclusion order $\subseteq$ is the order $\leq$ and so, if there exists a mask $M$ such that
$$\forall Y\in\overline{1,\,H},\,X\in\overline{1,\,W}\ \exists i\in\overline{Y-h+1,\,Y},\,j\in\overline{X-w+1,\,X}\ M^i_j=1$$
then since ${M}\leq{M^*}$ we also have
$$\forall Y\in\overline{1,\,H},\,X\in\overline{1,\,W}\ \exists i\in\overline{Y-h+1,\,Y},\,j\in\overline{X-w+1,\,X}\ {M^*}^i_j=1$$
and vice versa: if $M^*$ satisfies the later, then there exists at least one such mask $M$ (precisely $M^*$) which satisfies the former. Thus, the two definitions are indeed equivalent.
\qed\end{proof}

While from a formal standpoint the two definitions are equivalent, from a computational standpoint it is more convenient for us to work with the strong definition, since we do not have to consider all masks.

\begin{lemma}
\label{lem:maximal_mask}
Given two images $T\in\mathrm{Mat}_{H,\,W}\left(\Sigma\right)$ and $C\in\mathrm{Mat}_{h,\,w}\left(\Sigma\right)$ the construction of the maximal mask of $T$ with respect to $C$ takes $\Theta\left(WH\right)$ time.
\end{lemma}
\begin{proof}
Since the size of the output is $WH$ we have the lower bound $\Omega\left(WH\right)$. We effectively only have to prove the upper bound of $\mathcal{O}\left(WH\right)$.

We begin by studying the case $H=1$. In this case the maximal mask of $T$ with respect to $C$ consists of all occurrences of $C$ in $T$. This can be found in linear time, for example using the Knuth-Morris-Pratt algorithm (KMP~\cite{knuth1977fast}), with a runtime of $\mathcal{O}\left(W+w\right)$, which is $\mathcal{O}\left(WH\right)$ since $H=h=1$ and $w<W$.

For the case $H\neq 1$, we look for a two dimensional generalization of the Knuth-Morris-Pratt algorithm. One such generalization is Bird's algorithm~\cite{bird1977two} which uses KMP and a generalization of it due to Aho and Corasick~\cite{aho1975efficient} to find the rows and then columns where the pattern occurs.  

The output of Bird's algorithm is the list of occurrences of $C$ in $T$, i.e. the pairs $\left(i,j\right)$ such that ${M^*}^i_j=1$. Consequently, we can recover $M^*$ by taking ${M^*}^i_j= stage\left(i+h-1,\,j+w-1\right)$. This yields the maximal mask in $\mathcal{O}\left(WH+wh\right)=\mathcal{O}\left(WH\right)$ time.
\qed\end{proof}

\begin{theorem}[Image Cover Decision]
Given two images $T\in\mathrm{Mat}_{H,\,W}\left(\Sigma\right)$ and $C\in\mathrm{Mat}_{h,\,w}\left(\Sigma\right)$ checking if $C$ is a cover of $T$ takes $\Theta\left(WH\right)$ time (Algorithm~\ref{alg:recognition}).
\end{theorem}

\begin{proof}

We can instantly disqualify images $C$ having $h > H$ or $w > W$. Otherwise, since we must at least read $T$, the decision problem is at least $\Omega\left(HW\right)$. Thus, we prove only the upper bound, $\mathcal{O}\left(WH\right)$.

By Lemma~\ref{lem:maximal_mask} we compute $M^*$ in $\mathcal{O}\left(WH\right)$ time. We now check if $M^*$ ``tiles up to'' $T$, as per Definition~\ref{def:strong}. Thus we check that every $\left(x,y\right)$ of $T$ belongs to some occurrence of $C$, whose north-west corner is located at some point in $D(x,y)$, where
$$ D\left(x,\,y\right)=\lbrace \left(x-w+1,\,y-h+1\right)\leq\left(x^\prime,\,y^\prime\right)\leq\left(x,\,y\right) \vert {M^*}^{y^\prime}_{x^\prime}=1\rbrace.$$

At this point we could simply walk through $M^*$ and check that every location is indeed covered. However since there are up to $\mathcal{O}\left(WH\right)$ occurences of $C$ in $T$ the naive aproach takes $\mathcal{O}\left(W^2H^2\right)$ time. 

For the rest of the proof we show that we can compute whether there exists some $\left(x,y\right)$ for which $D\left(x,y\right)=\emptyset$ in $\mathcal{O}\left(WH\right)$. We call points for which $D\left(x,y\right)\neq0$ admissible and points for which $D\left(x,y\right)= 0$ inadmissible. We say that the points $D\left(x,y\right)$ ``support'' the hypothesis that $\left(x,y\right)$ is admissible.

Let $\leq_{lex}$ be the lexicographical order and the function $N\left(x,\,y\right)$ be the closest (north-west corner of an) occurrence of $C$ form $\left(x,y\right)$, i.e. 
$$ N\left(x,\,y\right)=\underset{\leq_{lex}}{\arg\min}\lbrace\left(x-x^\prime,\,y-y^\prime\right)\vert \left(x^\prime,\,y^\prime\right)\in D\left(x,\,y\right)\rbrace,$$
for which, by definition, $N\left(x,\,y\right)=\left(\infty,\infty\right)$ if and only if $D\left(x,y\right)=\emptyset.$

Note that if the minimal support for the western neighbor of a point, $N\left(x-1,y\right)$, does not support it, then $\left(x,\,y\right)$ is the only point that can support itself but not its northern neighbor, $\left(x,y-1\right)$, i.e. 
\begin{align*}
    N\left(x-1,\,y\right)\not\in D\left(x,\,y\right)&\Rightarrow {M^*}^{y^\prime}_{x^\prime}=0\ \forall x^\prime \in \overline{x-w+1,\,x-1},\,y^\prime\in\overline{y-h+1,\,y}\Rightarrow\\
    &\Rightarrow D\left(x,\,y\right)\subseteq D\left(x,\,y-1\right)\cup\lbrace \left(x,\,y\right)\rbrace.
\end{align*}

Similarly, if the minimal support for the northern neighbor of a point, $N\left(x,y-1\right)$, does not support it, then $\left(x,\,y\right)$ is the only point that can support itself but not its western neighbor, $\left(x-1,\,y\right)$, i.e. 
\begin{align*}
  N\left(x,\,y-1\right)\not\in D\left(x,\,y\right)&\Rightarrow {M^*}^{y^\prime}_x=0\ \forall y^\prime\in\overline{y-h+1,\,y-1}\Rightarrow\\
  &\Rightarrow D\left(x,\,y\right)\subseteq D\left(x-1,\,y\right)\cup\lbrace \left(x,\,y\right)\rbrace.
\end{align*}

By the above, if neither minimal support for the western and northern neighbors supports $\left(x,y\right)$ then only $\left(x,\,y\right)$ may support itself, i.e.
\begin{align*}
   &N\left(x,\,y-1\right)\not\in D\left(x,\,y\right) ,\, N\left(x-1,\,y\right)\not\in D\left(x,\,y\right)\Rightarrow\\ 
   \Rightarrow&{M^*}^{y^\prime}_{x^\prime}=0\ \forall x^\prime \in \overline{x-w+1,\,x},\,y^\prime\in\overline{y-h+1,\,y},\,\left(x^\prime,\,y^\prime\right)\neq\left(x,\,y\right) \Rightarrow \\
   \Rightarrow & D\left(x,\,y\right)\subseteq\lbrace (x,\,y)\rbrace.
\end{align*}

Moreover, if $\left(x_1,\,y_1\right)\leq_{lex} \left(x_2,\,y_2\right)$ we have
\begin{align*}&\left(x_1-x_1^*,\,y_1-y_1^*\right)\leq_{lex} \left(x_1-x_1^\prime,\,y_1-y_1^\prime\right) \Leftrightarrow \\ \Leftrightarrow & \left(x_2-x^*,\,y_2-y^*\right)\leq_{lex} \left(x_2-x^\prime,\,y_2-y^\prime\right)\end{align*}
and thus, if $\left(x^\prime,\,y^\prime\right)$ supports both $\left(x,\,y\right)$ and one of its western or northern neighbors, but is not the minimal support of that neighbor, then it is not the minimal support of $\left(x,\,y\right)$. We obtain the dynamic programming scheme
$$ N\left(x,\,y\right)\in\lbrace N\left(x-1,\,y\right),\, N\left(x,\,y-1\right),\,\left(x,\,y\right)\rbrace.$$

This scheme can be implemented in $\mathcal{O}\left(WH\right)$ time as shown in Algorithm \ref{alg:recognition}. We have proven that it correctly decides whether the maximal mask does indeed cover the entire image, i.e. $C$ is a cover of $T$. We conclude that the complexity of the decision problem is indeed $\Theta\left(WH\right)$. \qed
\end{proof}

\begin{algorithm}[!ht]
\caption{Image Cover Decision}
\label{alg:recognition}
\begin{algorithmic}[1]
\Procedure {Check}{$T$, $w$, $h$}
\State Preprocess $T$ (per Bird's algorithm)
\For {$x\in\overline{1,\,H}$}
\For {$y\in\overline{1,\,W}$}
\State $N\left(x,\,y\right)=\left(-\infty,\,-\infty\right)$
\If {$x > 1$ and $\left(x-w+1,\,y-h+1\right)\leq N\left(x-1,\,y\right)$}
\State $N\left(x,\,y\right)=N\left(x-1,\,y\right)$
\EndIf
\If {$y > 1$ and $\left(x-w+1,\,y-h+1\right)\leq N\left(x,\,y-1\right)$}
\If {$\left(x,\,y\right)-N\left(x,\,y-1\right)\leq_{lex} N\left(x,\,y\right)-N\left(x,\,y-1\right)$}
\State $N\left(x,\,y\right)=N\left(x,\,y-1\right)$
\EndIf
\EndIf
\If {$\textmd{stage}\left(y,\,x\right)$ (per Bird's algorithm)}
\State $N\left(x,\,y\right)=\left(x,\,y\right)$
\EndIf
\If {$N\left(x,\,y\right)=\left(-\infty,\,-\infty\right)$}
\State \Return Mismatch: $\left(x,\,y\right)$
\EndIf
\EndFor
\EndFor
\State \Return Match
\EndProcedure
\end{algorithmic}
\end{algorithm}

\section{Minimal Image Covers}
\label{sec:minimal_image_covers}

Among the family of covers of an image $T$, our goal is to find a ``minimal'' one. To achieve this goal we have to define the optimization criterion. This criterion takes the form of an evaluation function:
$ eval\,:\,\overline{1,\,W}\times\overline{1,\,H}\rightarrow \bar{\mathbb{R}}$.

\begin{proposition}Obtaining the minimal image cover $C$ of $T$ with respect to $eval$ takes time $\mathcal{O}\left(W^2H^2 + WH\Theta\left(eval\right) \right)$.
\end{proposition}
\begin{proof}
A brute force approach checks all possible $\left(w,h\right)$ pairs (which are $\Theta\left(WH\right)$) and uses the decision algorithm above. If a cover is found it is evaluated. This yields complexity $\mathcal{O}\left(W^2H^2 + WH\Theta\left(eval\right) \right)$.

Moreover, if $eval$ is arbitrary all $\left(w,h\right)$ pairs must be checked since $eval\left(.\right)$ can be unbounded (or some large finite value) for all $\left(w,h\right)$ except $\left(w^*,h^*\right)$ which shows the bound is tight.
\qed\end{proof}

\begin{proposition}If the minimal $C$ is the $n$-th candidate according to the order induced by $eval$, we can obtain $C$ in $\mathcal{O}\left(nWH\right)$ if the input is already sorted according to this order.
\end{proposition}

\begin{proof}
If the ordering induced by the $eval$ function is known, we queue up the would-be covers in that order (by sorting for example). For instance, if the $n$-th candidate is the first cover encountered, the runtime of the minimization algorithm described above is $\mathcal{O}\left(WH\Theta(eval) + nWH\right)$. This can be achieved via sorting, yielding a complexity of
 $$\mathcal{O}(WH\Theta(eval) + WH\log(WH) + nWH)$$
 \qed\end{proof}

\subsection{The Size Criteria}

We now study minimality with respect to a natural criterion, namely the size, as given by the area, $\ell_1$ norm  and $\ell_\infty$ norm. 

\noindent For the area, the evaluation function is defined as
$$ \overline{1,\,W}\times\overline{1,\,H}\ni\left(w,\,h\right)\rightarrow eval\left(w,\,h\right)=wh\in\mathbb{R}.$$

Suppose we knew that $wh\leq w_0h_0$. Then we have
$$h\in\overline{1,\,\min\left(\lfloor w_0h_0/w\rfloor,\,H\right)},$$
and thus $\left(w,h\right)$ is one of the lattice points of the intersection of the rectangle $\left(\left(1,\,1\right),\,\left(1,\,H\right),\,\left(W,\,1\right),\,\left(W,\,H\right) \right)$ with the triangle $\left(\left(1,\,1\right),\,\left(1,\,w_0h_0\right),\,\left(w_0h_0,\,1\right)\right)$. 

These contain at most $WH$ and $w_0^2h_0^2/2$ lattice points respectively. Thus, the optimal $\left(w,h\right)$ pair is found after at most  $n\approx \mathcal{O}\left(\min\left(w_0^2h_0^2,\,WH\right)\right)$ attempts which leads to an upper bound of $\mathcal{O}\left(\min\left(w_0^2h_0^2,\,WH\right)WH\right)$. 

\vspace*{1em}

\noindent For the $\ell_1$ norm the evaluation function is
$$ \overline{1,\,W}\times\overline{1,\,H}\ni\left(w,\,h\right)\rightarrow eval\left(w,\,h\right)=w+h\in\mathbb{R}.$$

Suppose we knew that $w+h\leq w_0+h_0$. Then we have 
$$h\in\overline{1,\,\min\left(w_0+h_0-w,\,H\right)},$$
and thus $\left(w,h\right)$ is one of the lattice points of the intersection of the rectangle $\left(\left(1,\,1\right),\,\left(1,\,H\right),\,\left(W,\,1\right),\,\left(W,\,H\right)\right)$ with the triangle $\left(\left(1,\,1\right),\,\left(1,\,w_0+h_0-1\right),\,\left(w_0+h_0-1,\,1\right)\right)$. 

These contain at most $WH$ and $\left(w_0+h_0\right)^2/2$ lattice points respectively. Thus, the optimal $\left(w,h\right)$ pair is found after at most $n\approx \mathcal{O}\left(\min\left(w_0^2+h_0^2,\,WH\right)\right)$ attempts which leads to an upper bound of $\mathcal{O}\left(\min\left(w_0^2+h_0^2,\,WH\right)WH\right)$. 

\vspace*{1em}

\noindent For the $\ell_\infty$ norm the evaluation function is defined as
$$ \overline{1,\,W}\times\overline{1,\,H}\ni\left(w,\,h\right)\rightarrow eval\left(w,\,h\right)=\max\left(w,\,h\right)\in\mathbb{R}.$$

Suppose we knew that $\max\left(w,\,h\right)\leq \max\left(w_0,\,h_0\right)$. Then we have $$h\in\overline{1,\,\min\left(\max\left(w_0,h_0\right),\,H\right)},$$
and thus $\left(w,h\right)$ is one of the lattice points of the intersection of the rectangle $\left(\left(1,\,1\right),\,\left(1,\,H\right),\,\left(W,\,1\right),\,\left(W,\,H\right)\right)$ with the square $\left(\left(1,\,1\right),\,\left(1,\,\max\left(w_0,h_0\right),\right),\,\left(\max\left(w_0,\,h_0\right),\,1\right),\,\left(\max\left(w_0,\,h_0\right),\,\max\left(w_0,\,h_0\right)\right)\right)$.

These contain at most $WH$ and $\max\left(w_0,\,h_0\right)^2$ lattice points respectively. Thus, the optimal $\left(w,h\right)$ pair is found after at most $n\approx \mathcal{O}\left(\min\left(\max\left(w_0,\,h_0\right)^2,\,WH\right)\right)$ attempts which leads to an upper bound of $\mathcal{O}\left(\min\left(\max\left(w,h\right)^2,\,WH\right)WH\right)$.

Note that we never used $w_0$ or $h_0$ other than for the calculation of the algorithm runtime. Thus, these calculations remain valid even if we do not know anything about $w_0$ and $h_0$. Their value is automatically substituted for the width and height of the minimal cover.

\subsection{Boosting Average Performance by Preprocessing}

In many cases, we do not have to verify all candidates. For instance, if the candidate a $\left(w,\,h\right)$ is a cover, then the first and the last $w$ columns and $h$ rows are image-covered by $T^{\overline{1,\,w}}_{\overline{1,\,h}}$. Based on this criterion we construct a preprocessing routine.

Suppose we knew that $h\geq h_0$. This means that $C$ is at least $h_0$-tall. Hence, $T^{\overline{1,\,W}}_{\overline{1,\,h_0}}$ covers $T^{\overline{1,\,W}}_{\overline{1,\,H}}$. Note that ${M^*}^i_j=0\forall j\geq 2$ since there is not enough space to accomodate another tile horizontally. Consequently, $T^{\overline{1,w}}_i$ covers $T^{\overline{1,W}}_i$ for all $i\leq h_0$ or $i\geq H-h_0$. Thus if $C_i$ are all the covers of $T^{\overline{1,W}}_i$ then 
$$ w\geq \min 
\left\{\left|c\right|\mid c \in C_i, \forall i\in\overline{1,h_0}\cup\overline{H-h_0,H} \right\}$$ 

This bound can be calculated in $\mathcal{O}\left(Wh_0\right)$. Since $C$ is a cover of $T$ if and only if $C^T$ is a cover of $T^T$, if we knew that $w_C\geq w_0$ then $h_{C^\prime}\geq w_0$ and hence we can similarly obtain a lower bound for $w_{C^\prime}=h_C$ in $\mathcal{O}\left(Hw_0\right)$. 

Suppose we knew that $w\geq w_0$ and $h\geq h_0$. It takes $\mathcal{O}\left(Wh_0+w_0H\right)$ to check that this first test does not already disprove the eligibility of $\left(w_0,h_0\right)$. Notably, the covers of $T^{\overline{1,W}}_i$ and $T^i_{\overline{1,H}}$ can be pre-computed (or cached) such that the cumulative preprocessing time is $\mathcal{O}\left(WH\right)$, which is essentially free.

Since we have established that this preprocessing is effectively free we can do it entirely \emph{a priori}, i.e. obtain the transitive closure of the preprocessing function. Let $S$ be the matrix of string covers returned by the optimal Minimal String Cover algorithm for each line and $S^\prime$ for columns, i.e. $S^i_j=1$ if the first $j$ characters on the $i$-th line cover the $i$-th line and ${S^\prime}^i_j=1$ if the first $i$ characters on the $j$-th column cover the $j$-th column. The current preprocessing is equivalent to computing the Hadamard product of the matrices
\begin{align*}
{S_1}^i_j &=\min\left(S^i_j,\,{S_1}^{i-1}_j\right) \\
{S^\prime_1}^i_j &=\min\left({S^\prime}^i_j,\,{S^\prime_1}^{i}_{j-1}\right) \\
S^* &=\min\left(S_1,\,S^\prime_1\right)=S_1\odot S^\prime_1.
\end{align*}

Notably, the number of elements that are not pruned is the number of non-zero elements of $S^*$. However 

\begin{align*}
{S_1}^i_j &= \underset{i^\prime=1}{\overset{i}\prod} {S}^{i^\prime}_j \\
{S^\prime_1}^i_j &= \underset{j^\prime=1}{\overset{j}\prod} {S^\prime}^{i}_{j^\prime} \\
{S^*}^i_j &= {S_1}^i_j{S^\prime}^i_j=\underset{i^\prime=1}{\overset{i}\prod}\underset{j^\prime=1}{\overset{j}\prod}{S}^{i^\prime}_j{S^\prime}^{i}_{j^\prime}
\end{align*}

We now check the effectiveness of our preprocesing.

\begin{proposition}
Computing the matrix $S^*$ reduces the number of candidates that need to be checked to $\Theta\left(1\right)$ average time for arbitrary $H$ and $\Theta\left(1\right)$ worst-case for $H=1$.
\end{proposition}

\begin{proof}
Assume that there is a $p$ probability for any tile in $S^i_j$ and ${S^\prime}^i_j$ to be 1, and even the additional condition that $S^{mi}_j\geq S^{i}_j$ for all $m$ and assuming that there is no single-character line nor column. Then by the Euler approximation, the probability that ${S_1}^i_j$ be 1 is $p^{{i}/\log\left(i\right)}$, that ${S^\prime}^i_j$ be 1 is $p^{{j}/\log\left(j\right)}$ and thus the probability that ${S^*}^i_j$ be 1 is $p^{i/\log\left(i\right)+j/\log\left(j\right)}$. Thus the expected number of 1s, considering that $S^i_W=S^H_j=1$, is
$$\left(1+\underset{i=2}{\overset{H}\sum}p^{i/\log\left(i\right)}\right)\left(1+\underset{j=2}{\overset{W}\sum}p^{j/\log\left(j\right)}\right)\leq 
\left(1+\frac{p}{1-p}\right)^2=\frac{1}{\left(1-p\right)^2}.$$
\end{proof}

We conclude that there exists a solution that is linear on the average case, $\mathcal{O}\left(WH\right)$ and quadratic in the worst, with the output-sensitive complexity: $\mathcal{O}\left(whWH\right)$, but which reduces to $\mathcal{O}\left(W\right)$ for the 1-dimensional case.

\section{A Connection with Lattices}
\label{sec:lattices}
A lattice~\cite{wlodawer2013protein} is an additive subgroup $\mathcal{L}$ of $\mathbb{R}^n$ isomorphic to $\mathbb{Z}^n$. By definition, it is infinite and yet it is generated by $n$ elements. Consider the isomorphism $\phi:\mathbb{Z}^n\rightarrow \mathcal{L}$. The projection of the unit volume $\lbrace0,\,1\rbrace^n$ through this isomorphism $\phi\left(\lbrace0,\,1\rbrace^n\right)$ is called the primitive cell of the lattice and it can be tiled by translations to form the entire $\mathcal{L}$. Note that by isomorphism we have: 
$$\phi\left(\underset{i=1}{\overset{n}\sum} \lambda_i\mathbf{e}_i\right)=\underset{i=1}{\overset{n}\sum}\lambda_i\phi\left(\mathbf{e}_i\right)$$

Moreover, if $\mathcal{L}$ is a lattice, $R$ is a rotation and $S$ is a scaling matrix i.e. $S_i^j=0\Leftrightarrow i\neq j$ then $SR\mathcal{L}$ is isomorphic to $\mathcal{L}$ and thus when classifying lattices we can assume that there exists some $\phi\left(\mathbf{e}_i\right)=\mathbf{e}_1$. Moreover since $\mathbb{Z}^n$ is isomorphic to itself by the maps $\mathbf{e}_i\rightarrow \mathbf{e}_{\sigma\left(i\right)}$ for any permutation $\sigma\in S_n$, we can assume that $\phi\left(\mathbf{e}_1\right)=\mathbf{e}_1$. Thus, all 2-dimensional latices can be characterized by the relative phase and length of the second vector. 

\begin{figure}[!ht]
\noindent\begin{minipage}{\textwidth}
\begin{minipage}[c][3cm][c]{\dimexpr0.33\textwidth-5pt\relax}
\resizebox{\textwidth}{!}{
\begin{tikzpicture}
\filldraw [color={rgb,255:red,220; green,220; blue,255}, ] (1.5, 1.5) rectangle (2.5,2.5);
\foreach \i in {1,...,5}
	\foreach \j in {1,...,5} {
		\node at (\i,\j) {\textbullet};
		\node at (\i,\j+1) {\textbullet};
		\node at (\i+1,\j+1) {\textbullet};
		\node at (\i+1,\j) {\textbullet};
        \draw (\i,\j) -- (\i,\j+1);
        \draw (\i+1,\j+1) -- (\i,\j+1);
        \draw (\i+1,\j+1) -- (\i+1,\j);
        \draw (\i,\j) -- (\i+1,\j);
    }
\node[color=blue] at (1.5, 1.5) {\textbullet};
\draw[->,color=blue,very thick] (1.5, 1.5) -- (2.5, 1.5);
\draw[->,color=blue,very thick] (1.5, 1.5) -- (1.5, 2.5);
\draw[color=blue,very thick,dotted] (2.5,1.5) -- (2.5,2.5);
\draw[color=blue,very thick,dotted] (1.5,2.5) -- (2.5,2.5);
\end{tikzpicture}
}
\end{minipage}\hfill
\begin{minipage}[c][3cm][c]{\dimexpr0.33\textwidth-5pt\relax}
\resizebox{\textwidth}{!}{
\begin{tikzpicture}
\filldraw [color={rgb,255:red,220; green,220; blue,255}, ] ({sqrt(3)},2) -- ({2*sqrt(3)},2) -- ({2.5*sqrt(3)},3.5) -- ({1.5*sqrt(3)},3.5) -- ({sqrt(3)},2);
\node at ({sqrt(3)/2},-0.5) {\textbullet};
\draw ({sqrt(3)/2},-0.5) -- ({sqrt(3)},0);
\node at ({sqrt(3)*5.5},-0.5) {\textbullet};
\draw ({sqrt(3)*5.5},-0.5) -- ({sqrt(3)*5},0);

\node at ({sqrt(3)/2},7.5) {\textbullet};
\draw ({sqrt(3)/2},7.5) -- ({sqrt(3)},7);
\node at ({sqrt(3)*5.5},7.5) {\textbullet};
\draw ({sqrt(3)*5.5},7.5) -- ({sqrt(3)*5},7);

\foreach \r in {0,...,1}{
    \def \j {3*\r};
	\foreach \k in {1,...,4} 
    {
        \def \i {\k * sqrt(3)};
		\node at ({\i},\j) {\textbullet};
		\node at ({\i},\j+1) {\textbullet};
		\node at ({\i+sqrt(3)/2},\j+1.5) {\textbullet};
		\node at ({\i+sqrt(3)},\j+1) {\textbullet};
		\node at ({\i+sqrt(3)},\j) {\textbullet};
		\node at ({\i+sqrt(3)/2},\j-0.5) {\textbullet};
        \draw ({\i},\j) -- ({\i},\j+1);
        \draw ({\i+sqrt(3)/2},\j+1.5) -- ({\i},\j+1);
        \draw ({\i+sqrt(3)/2},\j+1.5) -- ({\i+sqrt(3)},\j+1);
        \draw ({\i+sqrt(3)},\j) -- ({\i+sqrt(3)},\j+1);
        \draw ({\i+sqrt(3)},\j) -- ({\i+sqrt(3)/2},\j-0.5);
        \draw ({\i},\j) -- ({\i+sqrt(3)/2},\j-0.5);
    };
    \def \j {3*\r+1.5};
	\foreach \k in {1,...,5} 
    {
        \def \i {\k * sqrt(3) - sqrt(3)/2};
		\node at ({\i},\j) {\textbullet};
		\node at ({\i},\j+1) {\textbullet};
		\node at ({\i+sqrt(3)/2},\j+1.5) {\textbullet};
		\node at ({\i+sqrt(3)},\j+1) {\textbullet};
		\node at ({\i+sqrt(3)},\j) {\textbullet};
		\node at ({\i+sqrt(3)/2},\j-0.5) {\textbullet};
        \draw ({\i},\j) -- ({\i},\j+1);
        \draw ({\i+sqrt(3)/2},\j+1.5) -- ({\i},\j+1);
        \draw ({\i+sqrt(3)/2},\j+1.5) -- ({\i+sqrt(3)},\j+1);
        \draw ({\i+sqrt(3)},\j) -- ({\i+sqrt(3)},\j+1);
        \draw ({\i+sqrt(3)},\j) -- ({\i+sqrt(3)/2},\j-0.5);
        \draw ({\i},\j) -- ({\i+sqrt(3)/2},\j-0.5);
    }
}

\foreach \r in {2}{
    \def \j {3*\r};
	\foreach \k in {1,...,4} 
    {
        \def \i {\k * sqrt(3)};
		\node at ({\i},\j) {\textbullet};
		\node at ({\i},\j+1) {\textbullet};
		\node at ({\i+sqrt(3)/2},\j+1.5) {\textbullet};
		\node at ({\i+sqrt(3)},\j+1) {\textbullet};
		\node at ({\i+sqrt(3)},\j) {\textbullet};
		\node at ({\i+sqrt(3)/2},\j-0.5) {\textbullet};
        \draw ({\i},\j) -- ({\i},\j+1);
        \draw ({\i+sqrt(3)/2},\j+1.5) -- ({\i},\j+1);
        \draw ({\i+sqrt(3)/2},\j+1.5) -- ({\i+sqrt(3)},\j+1);
        \draw ({\i+sqrt(3)},\j) -- ({\i+sqrt(3)},\j+1);
        \draw ({\i+sqrt(3)},\j) -- ({\i+sqrt(3)/2},\j-0.5);
        \draw ({\i},\j) -- ({\i+sqrt(3)/2},\j-0.5);
    };
}

\node[color=blue] at ({sqrt(3)}, 2) {\textbullet};
\draw[->,color=blue,very thick] ({sqrt(3)}, 2) -- ({2*sqrt(3)}, 2);
\draw[->,color=blue,very thick] ({sqrt(3)}, 2) -- ({1.5*sqrt(3)}, 3.5);
\draw[color=blue,very thick,dotted] ({2*sqrt(3)},2) -- ({2.5*sqrt(3)},3.5);
\draw[color=blue,very thick,dotted] ({1.5*sqrt(3)},3.5) -- ({2.5*sqrt(3)},3.5);

\end{tikzpicture}
}
\end{minipage}\hfill
\begin{minipage}[c][3cm][c]{\dimexpr0.33\textwidth-5pt\relax}
\resizebox{\textwidth}{!}{
\begin{tikzpicture}
\filldraw [color={rgb,255:red,220; green,220; blue,255}, ] ({1+sqrt(2)}, {1+sqrt(2)/2}) rectangle ({2*(1+sqrt(2))}, {2+1.5*sqrt(2)});
\foreach \r in {1,...,4}{
    \def \j {\r*(1+sqrt(2))};
	\foreach \k in {1,...,4} {
        \def \i {\k * (1+sqrt(2))};
		\node at ({\i},{\j}) {\textbullet};
		\node at ({\i},{\j+1}) {\textbullet};
		\node at ({\i+sqrt(2)/2},{\j+1+sqrt(2)/2}) {\textbullet};
		\node at ({\i+1+sqrt(2)/2},{\j+1+sqrt(2)/2}) {\textbullet};
		\node at ({\i+1+sqrt(2)},{\j+1}) {\textbullet};
		\node at ({\i+1+sqrt(2)},{\j}) {\textbullet};
		\node at ({\i+1+sqrt(2)/2},{\j-sqrt(2)/2}) {\textbullet};
		\node at ({\i+sqrt(2)/2},{\j-sqrt(2)/2}) {\textbullet};
        \draw ({\i},{\j}) -- ({\i},{\j+1});
        \draw ({\i+sqrt(2)/2},{\j+1+sqrt(2)/2}) -- ({\i},{\j+1});
        \draw ({\i+sqrt(2)/2},{\j+1+sqrt(2)/2}) -- ({\i+1+sqrt(2)/2},{\j+1+sqrt(2)/2});
        \draw ({\i+1+sqrt(2)},{\j+1}) -- ({\i+1+sqrt(2)/2},{\j+1+sqrt(2)/2});
        \draw ({\i+1+sqrt(2)},{\j+1}) -- ({\i+1+sqrt(2)},{\j});
        \draw ({\i+1+sqrt(2)/2},{\j-sqrt(2)/2}) -- ({\i+1+sqrt(2)},{\j});
        \draw ({\i+1+sqrt(2)/2},{\j-sqrt(2)/2}) -- ({\i+sqrt(2)/2},{\j-sqrt(2)/2});
        \draw ({\i},{\j}) -- ({\i+sqrt(2)/2},{\j-sqrt(2)/2});
    }
	\foreach \k in {1,...,5} {
        \def \i {\k * (1+sqrt(2))};
		\node at ({\i},{\j}) {\textbullet};
		\node at ({\i-sqrt(2)/2},{\j-sqrt(2)/2}) {\textbullet};
		\node at ({\i},{\j-sqrt(2)}) {\textbullet};
		\node at ({\i+sqrt(2)/2},{\j-sqrt(2)/2}) {\textbullet};
        \draw ({\i},{\j}) -- ({\i-sqrt(2)/2},{\j-sqrt(2)/2});
        \draw ({\i},{\j-sqrt(2)}) -- ({\i-sqrt(2)/2},{\j-sqrt(2)/2});
        \draw ({\i},{\j-sqrt(2)}) -- ({\i+sqrt(2)/2},{\j-sqrt(2)/2});
        \draw ({\i},{\j}) -- ({\i+sqrt(2)/2},{\j-sqrt(2)/2});
    }
}

\foreach \r in {5,...,5}{
    \def \j {\r*(1+sqrt(2))};
	\foreach \k in {1,...,5} {
        \def \i {\k * (1+sqrt(2))};
		\node at ({\i},{\j}) {\textbullet};
		\node at ({\i-sqrt(2)/2},{\j-sqrt(2)/2}) {\textbullet};
		\node at ({\i},{\j-sqrt(2)}) {\textbullet};
		\node at ({\i+sqrt(2)/2},{\j-sqrt(2)/2}) {\textbullet};
        \draw ({\i},{\j}) -- ({\i-sqrt(2)/2},{\j-sqrt(2)/2});
        \draw ({\i},{\j-sqrt(2)}) -- ({\i-sqrt(2)/2},{\j-sqrt(2)/2});
        \draw ({\i},{\j-sqrt(2)}) -- ({\i+sqrt(2)/2},{\j-sqrt(2)/2});
        \draw ({\i},{\j}) -- ({\i+sqrt(2)/2},{\j-sqrt(2)/2});
    }
}

\node[color=blue] at ({1+sqrt(2)}, {1+sqrt(2)/2}) {\textbullet};
\draw[->,color=blue,very thick] ({1+sqrt(2)}, {1+sqrt(2)/2}) -- ({2*(1+sqrt(2))}, {1+sqrt(2)/2});
\draw[->,color=blue,very thick] ({1+sqrt(2)}, {1+sqrt(2)/2}) -- ({1+sqrt(2)}, {2+1.5*sqrt(2)});
\draw[color=blue,very thick,dotted] ({2*(1+sqrt(2))},{1+sqrt(2)/2}) -- ({2*(1+sqrt(2))},{2+1.5*sqrt(2)});
\draw[color=blue,very thick,dotted] ({1+sqrt(2)},{2+1.5*sqrt(2)}) -- ({2*(1+sqrt(2))},{2+1.5*sqrt(2)});
\end{tikzpicture}
}
\end{minipage}

\begin{minipage}[c][1cm][t]{\dimexpr0.33\textwidth-5pt\relax}
\captionof{figure}{a grid lattice}
\end{minipage}\hfill
\begin{minipage}[c][1cm][t]{\dimexpr0.33\textwidth-5pt\relax}
\captionof{figure}{a hexagonal row lattice}
\end{minipage}\hfill
\begin{minipage}[c][1cm][t]{\dimexpr0.33\textwidth-5pt\relax}
\captionof{figure}{a mixed tiling which is actually a grid lattice}
\end{minipage}

\end{minipage}
\end{figure}

Given a volume in $n$-dimensional space and a lattice $\mathcal{L}\subseteq\mathbb{E}^n$, we can divide it according to the lattice i.e. given $\phi:\mathbb{Z}^n\rightarrow\mathcal{L}$ we have 
$$\mathbb{E}^n_\mathcal{L}=\lbrace C_\mathbf{l} = \overline{conv}\left(\lbrace \phi\left(\mathbf{l}+\mathbf{v}\right)\vert \mathbf{v}\in\lbrace 0,\,1\rbrace^n \rbrace\right) \vert \mathbf{l}\in\mathbb{Z}^n \rbrace$$

Note that the translation $\phi\left(\mathbf{l}\right)\rightarrow\phi\left(\mathbf{l}^\prime\right)$ maps $C_\mathbf{l}$ to $C_{\mathbf{l}^\prime}$ and thus the volume of any two cells is the same for a given $\mathcal{L}\subseteq\mathbb{E}^n$. Thus we can define the quantity $vol\left(\mathcal{L}\right) = vol\left(C_\mathbf{0}\right)$ to be the unit volume of a lattice $\mathcal{L}$.

Given some volumetric data $\mathbb{E}^n\supseteq \mathbf{V}\ni\mathbf{x}\rightarrow \psi\left(\mathbf{x}\right)\in\mathbb{R}$ we say that a lattice $\mathcal{L}$ is legal with respect to $\psi$ if $\psi$ is also translation invariant i.e.
$$ \psi\left(\phi\left(\mathbf{x}\right)\right)=\psi\left(\phi\left(\mathbf{x+v}\right)\right)\forall\mathbf{x}\in \mathbb{Z},\,\mathbf{v}\in\lbrace 0,\,1\rbrace^n,\,\phi\left(\mathbf{x}\right)\in \mathbf{V},\,\mathbf{x+v}\in\mathbf{V}$$

Moreover, $\mathcal{L}$ is natural with respect to $\psi$ if it is a legal lattice, minimal with respect to the unit volume. 

We would like to obtain the unit cell of the natural lattice given \emph{not the lattice points but instead a tiling of the unit cell that is cropped to a $W$-long, $H$-tall image that contains at least one copy of the unit cell} i.e. volumetric 2-dimensional data. 

Once we have found an unit cell, any translation or rotation of it is still an unit cell which describes the same geometry and thus we have no interest in selecting any particular one. We accept any unit cell of any natural lattice.

Since a legal lattice is invariant to translations, we may always fix the origin of one unit cell on $T_1^1$. Since it is invariant to rotations we may always fix that one of its unit vectors is along the $T^1$ row. However, it may be that the other axis is not along the $T_1$ column, as is the case for hexagonal lattices. Moreover, it may be that our image does not end after an integer number of tiles, but instead a fractional one. In this case, the end fraction has to appear in the cover. We conclude that the shortest cover may never contain more than the volume of the box-cover of 4 unit tiles. In fact, it never contains 2 entire unit tiles on any side. Moreover, it will always contain at least one unit tile or a seed of it. 

Note that this approach is especially interesting in the case of quasi-periodic crystals (which do not admit a Bravais lattice)~\cite{wlodawer2013protein}. This extends the $k$-covers problem~\cite{Iliopoulos98} and asks for the $k$ unit cells which have been used, for example in a Penrose tiling~\cite{bursill1985penrose}.

\section{Applications in Computer Graphics}
\label{sec:graphics}
Consider the task of producing huge, unique maps for games, such as mazes or dungeons. Without procedural terrain generation this task is anything between infeasible and impossible, depending on the desired size and the available time and budget. Many games use Wang Tiles~\cite{Derouet-Jourdan17,kopf2006recursive} to produce huge maps (an interesting example is the Infamous game produced by Sucker Punch). They have recently garnered around them a very large community.

Wang tiles are formal systems visually modeled by square tiles with colors on each side. Two Wang tiles may only be tiled along an edge if the colors match. The most popular problems concerning them were: whether a set of Wang tiles can cover the plane and whether this can be done in a periodic way~\cite{jeandel2015aperiodic}.
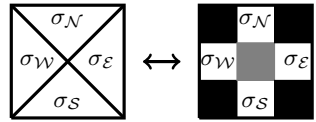
\begin{figure}[!ht]
\begin{center}
\begin{tikzpicture}
\draw[very thick] (0,0) -- (1.5,0) -- (1.5,1.5) -- (0,1.5) -- (0,0);
\draw[very thick] (0,0) -- (1.5,1.5);
\draw[very thick] (1.5,0) -- (0,1.5);
\node[anchor=north] at (0.75,1.5) {$\sigma_\mathcal{N}$};
\node[anchor=east] at (1.5,0.75) {$\sigma_\mathcal{E}$};
\node[anchor=south] at (0.75,0) {$\sigma_\mathcal{S}$};
\node[anchor=west] at (0,0.75) {$\sigma_\mathcal{W}$};

\draw[very thick, <->] (1.75, 0.75) -- (2.25, 0.75);

\draw[very thick] (2.5,0) -- (4,0) -- (4,1.5) -- (2.5,1.5) -- (2.5,0);
\filldraw (2.5,0) rectangle (3,0.5);
\filldraw (2.5,1) rectangle (3,1.5);
\filldraw[gray] (3,0.5) rectangle (3.5,1);
\filldraw (3.5,0) rectangle (4,0.5);
\filldraw (3.5,1) rectangle (4,1.5);
\node[anchor=north] at (3.25,1.5) {$\sigma_\mathcal{N}$};
\node[anchor=east] at (4.1,0.75) {$\sigma_\mathcal{E}$};
\node[anchor=south] at (3.25,0) {$\sigma_\mathcal{S}$};
\node[anchor=west] at (2.4,0.75) {$\sigma_\mathcal{W}$};
\end{tikzpicture}
\end{center}
\end{figure}

A Wang tile can also be represented as a 3-by-3 image. Two such images may be tiled together either along an edge or a corner. The formal system isomorphism is trivial: two 3-by-3 images may be tiled together on an edge if the respective colors on the Wang tiles match. This is very much like String Covers, except two such images may never be tiled one alongside another. 

Consider the following problems:

\begin{problem}[Minimal Wang Cover]
Given a tiling of some Wang Tiles check if there exists a periodic pattern covering it.
\end{problem}

\begin{problem}[k-Wang Covers]
Given a tiling of some Wang Tiles check if there exist k patterns which, when tiled together cover the image.
\end{problem}

\begin{problem}[Approximate Wang Cover]
Given a tiling of some Wang Tiles find the minimal number of pixels to be changed for it to be covered by a single periodic pattern.
\end{problem}

When given a 3-tall image the first two collapse to vectorial String Cover and vectorial k-Covers. For the last one, we must also impose that the black and gray pixels which we added ourselves are never corrupted. Thus we impose that the distance between two tilings is infinite if a black or gray pixel is corrupted. Hence it is equivalent to the Approximate String Cover of Amir et al. with an almost-Hamming metric.

\begin{problem}[Generalization to pseudo-metrics]
\label{prob:generalization_pseudometrics}
Given a compression palette $\Gamma\subseteq\Sigma$ and an algorithm that is consistent with respect to the colors it replaces i.e. $\mathcal{A}:\Sigma\rightarrow \Gamma$ and a tiling of some Wang Tiles, check if the solutions to the above problems change.
\end{problem}

The last problem is not important from a computational perspective; in fact it is quite trivial, but it gives substance to the pseudo-metric variations of String Cover problems.

Given computationally efficient algorithms that solve Problem~\ref{prob:generalization_pseudometrics}, there are several interesting applications in computer aided design (see e.g.~\cite{Derouet-Jourdan17}). 
One use of Wang tiles is procedure terrain generation in video games. If a player knows that the game he is playing uses Wang tiles, he can use an image cover algorithm to predict the next challenge.  
Another application is image compression: we can use these algorithms on images produced by designers in order to extract textures or motifs. 

Consider a game with hexagonal tiles that wants to make use of Perlin noise~\cite{perlin2002improving}. It is unnatural that it be used purely, since the rectangular lattice is not actually legal. On the other hand, since we can obtain $\mathcal{L}$, we by default have a mapping $\phi^{-1}:\mathcal{L}\rightarrow\mathbb{Z}^n$. In this domain, our lattice is indeed rectangular. Thus, it is here that we should apply our Perlin noise.

\begin{definition}
Given a lattice $\mathbb{Z}^n\overset{\phi}\rightarrow \mathcal{L}\subseteq\mathbb{E}^n$ and a noise-function appropriate for rectangular latices $\mathcal{P}:\mathbb{Z}^n\rightarrow \mathbb{R}$, we can lift it to $\mathcal{L}$:
$$ \mathcal{L}\ni\mathbf{x}\rightarrow \mathcal{P}_\mathcal{L}\left(\mathbf{x}\right) = \mathcal{P}\left(\phi^{-1}\left(\mathbf{x}\right)\right)\in\mathbb{R} $$
\end{definition}

Thus we can define the Perlin noise appropriate for a given Wang system. Note that the magnitude of Perlin noise is an input parameter. Thus, without changing the game or inducing unnatural patterns, as a game developer we can easily add a diversity grade for games using Wang tiles for terrain generation.

\bibliographystyle{abbrv}
\bibliography{bibliography2}

\begin{thebibliography}{10}

\bibitem{aho1975efficient}
A.~V. Aho and M.~J. Corasick.
\newblock Efficient string matching: an aid to bibliographic search.
\newblock {\em Communications of the ACM}, 18(6):333--340, 1975.

\bibitem{AmirLLLP17}
A.~Amir, A.~Levy, M.~Lewenstein, R.~Lubin, and B.~Porat.
\newblock Can we recover the cover?
\newblock In {\em 28th Annual Symposium on Combinatorial Pattern Matching,
  {CPM} 2017, July 4-6, 2017, Warsaw, Poland}, pages 25:1--25:15, 2017.

\bibitem{AmirLLP17}
A.~Amir, A.~Levy, R.~Lubin, and E.~Porat.
\newblock Approximate cover of strings.
\newblock In {\em 28th Annual Symposium on Combinatorial Pattern Matching,
  {CPM} 2017, July 4-6, 2017, Warsaw, Poland}, pages 26:1--26:14, 2017.

\bibitem{ApostolicoB97}
A.~Apostolico and D.~Breslauer.
\newblock Of periods, quasiperiods, repetitions and covers.
\newblock In {\em Structures in Logic and Computer Science, {A} Selection of
  Essays in Honor of Andrzej Ehrenfeucht}, pages 236--248, 1997.

\bibitem{ApostolicoFI91}
A.~Apostolico, M.~Farach, and C.~S. Iliopoulos.
\newblock Optimal superprimitivity testing for strings.
\newblock {\em Inf. Process. Lett.}, 39(1):17--20, 1991.

\bibitem{bacciotti2006liapunov}
A.~Bacciotti and L.~Rosier.
\newblock {\em Liapunov functions and stability in control theory}.
\newblock Springer Science \& Business Media, 2006.

\bibitem{bird1977two}
R.~S. Bird.
\newblock Two dimensional pattern matching.
\newblock {\em Information Processing Letters}, 6(5):168--170, 1977.

\bibitem{brodzik2007quaternionic}
A.~K. Brodzik.
\newblock Quaternionic periodicity transform: an algebraic solution to the
  tandem repeat detection problem.
\newblock {\em Bioinformatics}, 23(6):694--700, 2007.

\bibitem{bursill1985penrose}
L.~Bursill and P.~J. Lin.
\newblock Penrose tiling observed in a quasi-crystal.
\newblock {\em Nature}, 316(6023):50--51, 1985.

\bibitem{Derouet-Jourdan17}
A.~Derouet{-}Jourdan, M.~Salvati, and T.~Jonchier.
\newblock Procedural wang tile algorithm for stochastic wall patterns.
\newblock {\em CoRR}, abs/1706.03950, 2017.

\bibitem{HAVLIN1995171}
S.~Havlin, S.~Buldyrev, A.~Goldberger, R.~Mantegna, S.~Ossadnik, C.-K. Peng,
  M.~Simons, and H.~Stanley.
\newblock Fractals in biology and medicine.
\newblock {\em Chaos, Solitons and Fractals}, 6:171 -- 201, 1995.
\newblock Complex Systems in Computational Physics.

\bibitem{Iliopoulos98}
C.~Iliopoulos and W.~Smith.
\newblock An on-line algorithm of computing a minimum set of k-covers of a
  string.
\newblock In {\em Proc. of Ninth Australian Workshop on Combinatorial
  Algorithms (AWOCA)}, pages 97--106, 1998.

\bibitem{jeandel2015aperiodic}
E.~Jeandel and M.~Rao.
\newblock An aperiodic set of 11 wang tiles.
\newblock {\em arXiv preprint arXiv:1506.06492}, 2015.

\bibitem{katok1997introduction}
A.~Katok and B.~Hasselblatt.
\newblock {\em Introduction to the modern theory of dynamical systems},
  volume~54.
\newblock Cambridge university press, 1997.

\bibitem{knuth1977fast}
D.~E. Knuth, J.~H. Morris, Jr, and V.~R. Pratt.
\newblock Fast pattern matching in strings.
\newblock {\em SIAM journal on computing}, 6(2):323--350, 1977.

\bibitem{kopf2006recursive}
J.~Kopf, D.~Cohen-Or, O.~Deussen, and D.~Lischinski.
\newblock Recursive wang tiles for real-time blue noise.
\newblock {\em ACM Trans. Graph.}, 25(3):509--518, July 2006.

\bibitem{middlestead2017digital}
R.~Middlestead.
\newblock {\em Digital Communications with Emphasis on Data Modems: Theory,
  Analysis, Design, Simulation, Testing, and Applications}.
\newblock Wiley, 2017.

\bibitem{ming1990kolmogorov}
L.~Ming and P.~M. Vit{\'a}nyi.
\newblock Kolmogorov complexity and its applications.
\newblock In {\em Algorithms and Complexity}, pages 187--254. Elsevier, 1990.

\bibitem{muchnik2003almost}
A.~Muchnik, A.~Semenov, and M.~Ushakov.
\newblock Almost periodic sequences.
\newblock {\em Theoretical Computer Science}, 304(1-3):1--33, 2003.

\bibitem{perlin2002improving}
K.~Perlin.
\newblock Improving noise.
\newblock {\em ACM Trans. Graph.}, 21(3):681--682, July 2002.

\bibitem{searle1980speech}
J.~R. Searle, F.~Kiefer, M.~Bierwisch, et~al.
\newblock {\em Speech act theory and pragmatics}, volume~10.
\newblock Springer, 1980.

\bibitem{sethares1999periodicity}
W.~A. Sethares and T.~W. Staley.
\newblock Periodicity transforms.
\newblock {\em IEEE transactions on Signal Processing}, 47(11):2953--2964,
  1999.

\bibitem{timmermans2017cyclical}
M.~Timmermans, R.~Heijmans, and H.~Daniels.
\newblock Cyclical patterns in risk indicators based on financial market
  infrastructure transaction data, 2017.

\bibitem{tychonoff1935theoremesL}
A.~Tychonoff.
\newblock Th{\'e}or{\`e}mes d'unicit{\'e} pour l'{\'e}quation de la chaleur.
\newblock {\em Matematiceskij sbornik}, 42(2):199--216, 1935.

\bibitem{wlodawer2013protein}
A.~Wlodawer, W.~Minor, Z.~Dauter, and M.~Jaskolski.
\newblock Protein crystallography for aspiring crystallographers or how to
  avoid pitfalls and traps in macromolecular structure determination.
\newblock {\em The FEBS journal}, 280(22):5705--5736, 2013.

\end{thebibliography}

\end{document}